\documentclass[journal]{IEEEtran}
\usepackage{amsmath,paralist,amsthm,comment,color,amssymb,fancybox,framed}

\usepackage{flushend}
\usepackage{multicol}
\usepackage{times}
\usepackage{fullpage,amsfonts,latexsym,graphics,amssymb}
\usepackage{array}
\usepackage{amsmath,amsthm,amstext,url}
\usepackage[mathscr]{eucal}
\usepackage{cite}
\usepackage[hidelinks,colorlinks   = true,urlcolor     = blue, 
  linkcolor    = blue, 
  citecolor   = red 
  ]{hyperref}
  
\usepackage{tikz}
\usepackage{pgfplots}
\pgfplotsset{compat=1.14} 
\usepackage{float}
\usepackage{booktabs}
\usepackage{diagbox}

\usepackage{varwidth}
\usepackage{adjustbox,xcolor}
\usepackage{mathabx}

\usetikzlibrary{fit,shapes}

\newtheorem{theorem}{Theorem}

\newtheorem{lemma}[theorem]{Lemma}
\newtheorem{corollary}[theorem]{Corollary}

\theoremstyle{definition}

\theoremstyle{definition}

\newtheorem*{remark}{Remark}
\DeclareMathOperator{\rk}{rk}

\definecolor{shadecolor}{gray}{.95}

\title{Subpacketization in Coded Caching with Demand Privacy}
\author{\IEEEauthorblockN{Aravind V R\IEEEauthorrefmark{1} \qquad
Pradeep Sarvepalli\IEEEauthorrefmark{2} \qquad Andrew Thangaraj\IEEEauthorrefmark{3}}\\
\IEEEauthorblockA{Department of Electrical Engineering,
Indian Institute of Technology Madras, India\\
Email: \IEEEauthorrefmark{1}ee13d205@ee.iitm.ac.in,
\IEEEauthorrefmark{2}pradeep@ee.iitm.ac.in,
\IEEEauthorrefmark{3}andrew@ee.iitm.ac.in}}

\begin{document}

\maketitle

\begin{abstract}
    Coded caching is a technique where we utilize multi-casting opportunities to reduce rate in cached networks. One limitation of coded caching schemes is that they reveal the demands of all users to their peers. In this work, we consider coded caching schemes that assure privacy for user demands with a particular focus on reducing subpacketization. For the 2-user, 2-file case, we present a new linear demand-private scheme with the lowest possible subpacketization. This is done by presenting the scheme explicitly and proving impossibility results under lower subpacketization. Additionally, when only partial privacy is required, we show that subpacketization can be significantly reduced when there are a large number of files.
\end{abstract}

\section{Introduction}
Data traffic has been growing rapidly in recent years with content delivery, especially that of multimedia files, contributing a significant part. One important aspect of such traffic is its temporal variation. Network usage during peak demand times could be much higher than the demand in off-peak hours. Caching is a way to alleviate network congestion during peak hours by prefetching popular content nearer to the user during off-peak hours. Depending on the limitations on memory, a part of these files would be prefetched and once the user makes a demand, the rest of the requested file will be transmitted. Early literature on caching focused on cache placement/replacement policies\cite{aggarwal1999caching}, caching architectures\cite{chankhunthod1996hierarchical,michel1998adaptive,povey1997distributed}, web request models \cite{breslau1999web} etc.

Maddah-Ali and Niesen had shown in their seminal paper that coding can achieve significant gain over uncoded caching by making use of multicast opportunities \cite{maddah2014fundamental}. Coded caching achieves an additional \textit{global caching gain}, which is proportional to the number of users.
Their scheme is shown to be order optimal with an information-theoretic lower bound on the number of files needed to be transmitted (known as \textit{rate}). Though the exact lower bound on peak rate is still an open problem several works had investigated this and came up with tighter bounds \cite{sengupta2015improved,ghasemi2017improved,wan2016optimality,yu2017exact}. The problem has been studied in several settings like decentralized caching \cite{maddah2015decentralized}, non-uniform demands \cite{niesen2016coded}, multiple levels of cache \cite{karamchandani2016hierarchical} to name a few. Most of the schemes in these works involve storing the prefetched parts of files in  uncoded form. 
Coded prefetching is investigated in \cite{chen2016fundamental, tian2018caching, gomez2018fundamental}, where linear combinations of subfiles are stored in caches. 
In a few regimes this approach can improve the rate-memory trade-off over uncoded prefetching.

Yan \textit{et al.} developed a structure called \textit{placement delivery arrays} that could model both the placement and delivery schemes in a single array \cite{yan2017placement}. Graphical models for caching have been investigated in \cite{shanmugam2017coded, yan2017bipartite, shangguan2018centralized}. 
Schemes can also be derived using combinatorial designs and linear block codes \cite{tang2018coded}. A limitation with the original centralized scheme was the high subpacketization of files \cite{shanmugam2016finite}.
In the original scheme due to \cite{maddah2014fundamental}, the number of subfiles a file is split into,
increases exponentially with the number of users. 
These  combinatorial models  have helped in developing schemes that have lower subpacketization but with a small penalty on rate \cite{shanmugam2017coded, cheng2017coded}.

One area of particular interest is security and privacy in coded caching. In typical coded caching schemes, other users involved in the multicast or eavesdroppers might get to know the identity of the file a particular user demanded and its contents. 
Furthermore,  users will be able to partially access files which they have not demanded. 
This is in part due to the cache that contains contents of files not requested by them and also because, during delivery, they may be able to decode packets not meant for them.  
Sengupta \textit{et al.} \cite{sengupta2014fundamental} proposed a method for preventing information leakage to an external wiretapper with the use of cryptographic keys. Visakh \textit{et al.} \cite{ravindrakumar2016fundamental} had recently shown that the contents of a file could be revealed only to the user/users who requested it, using secret sharing techniques. 

One aspect that has not been investigated much is the privacy of the user requests in the specific context of coded caching, while it has been studied in closely related areas like index coding \cite{karmoose2017private} and private information retrieval (PIR) \cite{chor1995private}. As we were preparing this manuscript, we became aware of work due to Wan and Caire \cite{wan2019coded} who take a different approach for user request privacy from ours. Another paper by Kamath \cite{kamath2019demand} also addressed the problem of demand privacy and their approach is similar to the one in this work. We point out the specific differences in our results when compared to those from \cite{wan2019coded} and \cite{kamath2019demand} below.

In this work, we explore methods to obtain privacy of each user's requests from the other users in coded caching keeping subpacketization constraints as an important parameter. 

Our specific contributions are as follows:
\begin{compactenum}[i)]
\item We focus on the 2-user, 2-file case in detail and provide an achievable multicast transmission rate versus cache storage curve under a demand privacy constraint.
\item For the 2-user, 2-file case with cache storage of 1 file, we show an explicit demand-private scheme achieving a multicast transmission rate of 2/3 with a subpacketization of 3. This scheme cannot be obtained using the general scheme proposed in \cite{kamath2019demand}, which, in fact, requires a subpacketization of 6. 
\item For the 2-user, 2-file case, we prove some impossibility results on subpacketization of 2 and uncoded cache storage for linear coded caching with demand privacy. These are some of the first negative results in this new area.
\item Finally, we propose a general $K$-user, $N$-file partially demand-private scheme that provides a trade-off between the level of privacy and reduction in subpacketization.
\end{compactenum}

The rest of the paper is organized as follows.
In Section~\ref{sec:ps}, we describe the system setup and the problem statement. 
In Section~\ref{sec:22}, we provide demand-private schemes and an achievable rate vs cache memory curve for the case of two users and two files. We prove certain impossibility results with respect to packetization and coded prefetching. 
In Section~\ref{sec:gen-scheme}, we describe the general scheme for constructing demand-private coded caching schemes from non-private coded caching schemes from \cite{kamath2019demand}, and provide specific instances of the construction from PDAs resulting in lesser subpacketization. 
We also introduce the notion of partially private schemes and show how to construct a partially private scheme.
We conclude with a brief discussion on scope for future work in Section~\ref{sec:conc}. 

\section{Problem Statement}\label{sec:ps}
\subsection{System setup}
Assume that we have a server with $N$ files. 
Each file is assumed to be of $F$ bits and the $i$-th file is denoted $W_i$. 
The server is connected to $K$ users via a multicast link. 
Each user has a cache of size $MF$ bits.
The cache contents of the $i$-th user are denoted $Z_i$.
The system setup is shown in 
Fig.~\ref{fig:setup}.
\begin{figure}[htb]
     \centering
     \begin{tikzpicture}[>=stealth, thick]
        \draw [thick] (0.95,0) rectangle (2.05,1.8);				
        \draw[thick] (0.95,0.45)--(2.05,0.45);						
        \draw[thick] (0.95,0.9)--(2.05,0.9);						
        \draw[thick] (0.95,1.35)--(2.05,1.35);						
        \coordinate [below of=s,node distance=0cm] (s) at (1.5,0) {};
        \coordinate [below of=s,node distance=1.5cm] (b1) {};
        \node[draw, thick, align=center, inner xsep=2pt, minimum height=6mm, below left of=b1,node distance=1.2cm] (z2) {$Z_1$};
        \node[draw, thick, align=center, inner xsep=2pt, minimum height=6mm, left of=z2,node distance=1cm] (z1) {$Z_0$};
        \node[draw, thick, align=center, inner xsep=2pt, minimum height=6mm, right of=z2,node distance=2.8cm] (zk) {$Z_{K-1}$};

        \path[thick] (s) edge node[text width=1cm, right] {$X^D$} coordinate (m) (b1);
        \draw[shift={(m)}](-0.1,-0.1)--(0.1,+0.1);

        \draw[->,thick] (b1) -- (z1.north);
        \draw[->,thick] (b1) -- (z2.north);
        \draw[->,thick] (b1) -- (zk.north);
        \node [right] at (2.15,0.6) {Server};					
        \node [above] at (1.5,1.28) {$W_0$};					
        \node [above] at (1.5,0.83) {$W_1$};					
        \node [above] at (1.5,0.38) {$\vdots$};					
        \node [above] at (1.5,-0.07) {$W_{N-1}$};				
        \path (z2) -- node[auto=false]{\ldots} (zk);
     \end{tikzpicture}
     \caption{Caching system.}
     \label{fig:setup}
 \end{figure}
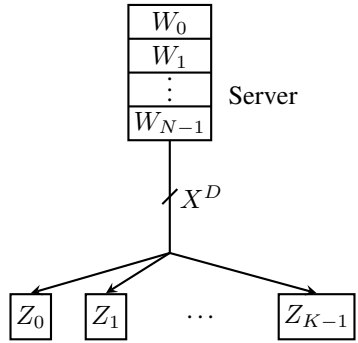
The cache system works in two phases. In the first phase called the \textit{placement phase}, the cache of each user is populated with content by the server. In addition, the server sends metadata or header information $\Theta(Z_i)$ about how the cache content was derived from the files to User~$i$. The header information is assumed to be small in size when compared to the file size but crucial for decoding purposes. Note that during the placement phase  the server is unaware of the files demanded by the users. We  assume that the transmission of cache content and header takes place over a private link between the server and each user.

 In the second phase, called the \textit{delivery phase}, each user requests the server for one  of the files from the set of $N$ files. 
 The demand of the $i$-th user is denoted $D_i$,  where $D_i \in [N]\triangleq\{0,1,\ldots,N-1\}$.
 The demands of all the users $0$ to $K-1$ is denoted by the demand vector $D= (D_0, D_1, \ldots, D_{K-1})$. 
 We  assume that the $D_i$ are all i.i.d. random variables uniformly distributed over
 $[N]$ and that the demands are sent over a private link between the user and the server.
 Based on the demands, the server multicasts $\ell$ packets, typically of the same size.
 The entire multicast transmission from the server is denoted $X^D$ for a demand vector $D$. It consists of $RF$ bits. The transmission $X^D$ depends on the cache $Z_i$ and the demands $D_i$. The quantity $R$ is called the rate of transmission. In addition to $X^D$, some additional metadata or header information about the transmission is typically multicast in coded caching schemes. This metadata, denoted $\Theta(X^D)$, is usually small compared to the file size and provides critical information for decoding by the users.
 
 The main requirement in a coded caching scheme is that User~$i$ should be able to decode the file $W_{D_i}$ using $Z_i$, $\Theta(Z_i)$, $X^D$ and $\Theta(X^D)$. In other words, we require 
 \begin{equation}
     H(W_{D_i}\;|\;Z_i,\Theta(Z_i), X^D, \Theta(X^D)) = 0.
 \end{equation}
 We denote a coded caching scheme with $K$
 users, $N$ files, local cache size $M$, 
 and rate $R$ as a $(K,N;M,R)$ coded caching scheme, or as a $(K,N)$ scheme in short.
 
 \subsection{Demand privacy in coded caching}

We will introduce the notion of demand privacy in coded caching with a simple example. Consider the  $(2,2)$ coded caching scheme due to Maddah-Ali and Niesen \cite{maddah2014fundamental}
shown in Fig.~\ref{fig:mn-2-2}.

    \begin{figure}[htb]
     \centering
     \begin{tikzpicture}[>=stealth, thick,baseline=-1.6cm]
        \node[draw, thick, align=center, inner xsep=20pt, minimum height=6mm] (s) {$A$\\$B$};
        \coordinate [below of=s,node distance=2cm] (b1) {};
        \coordinate [below of=b1,node distance=1cm] (b2) {};
        \node[draw, thick, align=center, inner xsep=2pt, minimum height=6mm, left of=b2,label=below:$Z_0$, node distance=1.2cm] (z2) {
            $A_{0}$, $B_{0}$
        };

        \node[draw, thick, align=center, inner xsep=2pt, minimum height=6mm, right of=b2,label=below:$Z_1$,node distance=1.2cm] (z3) {
            $A_{1}$,   $B_{1}$
        };

        \path[thick] (s) edge node[text width=3cm, right] {
        $X^D$ 
        } (b1);
        
        \draw[-,thick] (s.west) -- (s.east);
        \draw[->,thick] (b1) -- (z2.north);
        \draw[->,thick] (b1) -- (z3.north);
     \end{tikzpicture}%
     \begin{tabular}{cc}
        \toprule
       $D_0D_1$ & X \\ \midrule
        $AA$ &  $A_1 \oplus A_0$\\
        $AB$ &  $A_1 \oplus B_0$\\
        $BA$ &  $B_1 \oplus A_0$\\
        $BB$ &  $B_1 \oplus B_0$\\ \bottomrule
   \end{tabular}
     \caption{Non-private scheme from \cite{maddah2014fundamental} fo
     r $N=2$ files, $K=2$ users and demand vector, $D=(D_1, D_2)$.}
     \label{fig:mn-2-2}
 \end{figure}
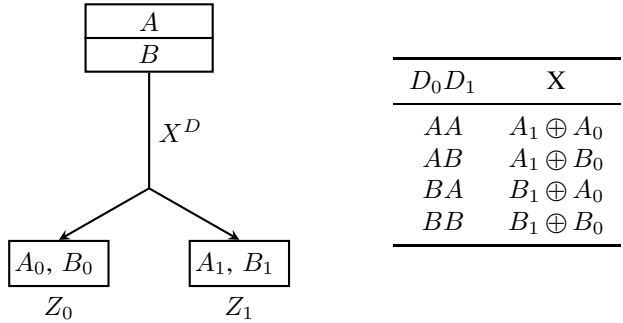

Suppose that the demand is $(A, A)$.  
This results in the transmission $A_1\oplus A_0$. To recover the files, each user must know what linear combination of subfiles has been transmitted. So, we will suppose that the server sends the linear combination information as header along with the transmission. 
It is easy to see that each user can recover the missing portion of the file demanded by them. 
However, the scheme has the unfortunate side effect of revealing the demands of each user to the other parties. From the header and scheme details, it is clear to User~0 that User~1 demanded the file $A$ and vice versa.

If the transmission is $A_i\oplus B_j$, then the $i$-th user  can infer that the $j$-th user has requested $A$ based on the linear combination header information. In general, users can use the combined information of their cache, demands and header data from the server to learn about another user's demands. 

Based on the preceding discussion, to achieve demand privacy in a coded caching  scheme, we impose the following additional condition for all demand vectors $D$: 
\begin{eqnarray}
I(D_i,Z_i,\Theta(Z_i),X^D,\Theta(X^D);D_j)=0,\quad i\ne j.
\label{eq:dpcond}
 \end{eqnarray}
 In other words, we require that the $i$-th user is completely uncertain about what the $j$-th user demands, given all information available to User~$i$ in the coded caching scheme. It can be shown that the standard Maddah-Ali-Niesen scheme \cite{maddah2014fundamental} does not satisfy the demand privacy condition in Eq.~\eqref{eq:dpcond}.

\section{$(K=2,N=2)$ coded caching with demand privacy}
\label{sec:22}
We will first consider the case when there are two files and two users. A complete characterization of the $M$ vs $R$ region in the case of two files/users was one of the starting points of the area of coded caching. Therefore, it is important to fully characterize the same region with demand privacy. We have made some partial progress towards this problem.

First, we will show the design of a linear $(2,2;1,2/3)$ coded caching scheme with demand privacy having a subpacketization (number of parts into which each file is divided) of 3. In comparison, directly converting a $(4,2)$-Maddah-Ali-Niesen scheme into a $(2,2;1,2/3)$ demand-private scheme requires a subpacketization of 6 \cite{kamath2019demand}. 

\subsection{$(M=1, R=2/3)$ scheme with subpacketization 3}
\label{sec:sub3}
The two files, $A$ and $B$, are divided into 3 parts $A_i$, $i=0,1,2$ and $B_i$, $i=0,1,2$. Table \ref{tab:schN2K2num1} summarizes the entire scheme. \begin{table}[H]
   \caption{$(K=2, N=2; M=1, R=2/3)$ demand-private caching scheme with subpacketization 3.
   If server assigns the cache $Z_{i0}$ to User~$i$, then $X^{D_0D_1}$ is the transmission for the demand $D_0D_1$. }
 \begin{tabular}{c m{5em}}
        \toprule
        Notation & Possible cache contents \\ \midrule
        $Z_{00}$ &  $
        \begin{array}{c}
	        A_0\oplus A_1 \\
	        B_0\oplus B_1 \\
	        A_2\oplus B_1
        \end{array}
        $\\\midrule[0.2pt]
        $Z_{01}$ &  $
        \begin{array}{c}
	        A_0\oplus A_1 \\
	        B_0\oplus B_1 \\
	        A_1\oplus B_2
        \end{array}
        $\\\midrule[0.2pt]
        $Z_{10}$ &  $
        \begin{array}{c}
	        A_0\oplus A_2 \\
	        B_0\oplus B_2 \\
	        A_1\oplus B_2
        \end{array}
        $ \\\midrule[0.2pt]
        $Z_{11}$ & $
        \begin{array}{c}
	        A_0 \oplus A_2 \\
	        B_0\oplus B_2 \\
	        A_2\oplus B_1
        \end{array}
        $\\ \bottomrule
   \end{tabular}  \qquad
    \begin{tabular}{cc}
        \toprule
        $D_0D_1$ & $X^{D_0D_1}$ \\ \midrule
        $AA$ &  $
        \begin{array}{c}
            A_0 \\ 
            B_0
        \end{array}
        $\\\midrule[0.2pt]
        $AB$ &  $
        \begin{array}{c}
            A_1 \\
            B_1 
        \end{array}
        $\\\midrule[0.2pt]
            $BA$ &  $
        \begin{array}{c}
            A_2 \\
            B_2 
        \end{array}
        $ \\\midrule[0.2pt]
        $BB$ & $
        \begin{array}{c}
            A_0\oplus A_1\oplus A_2 \\
            B_0\oplus B_1\oplus B_2
        \end{array}
        $\\ \bottomrule
    \end{tabular}
    \label{tab:schN2K2num1}
\end{table}
\begin{table}[htb]
\centering
\caption{Files recovered from possible cache pairs and transmission $X$ for the scheme from Table~\ref{tab:schN2K2num1}}
\resizebox{\columnwidth}{!}{%
\begin{tabular}{|c|cccc|}
\hline
    \diagbox{Caches\\$Z_{0}$,$Z_{1}$}{\raisebox{-1\height}{\ $X$}} 
    & $\begin{array}{c}
            A_0 \\ 
            B_0
        \end{array}$
    & $\begin{array}{c}
            A_1 \\ 
            B_1
        \end{array}$ 
    & $\begin{array}{c}
            A_2 \\ 
            B_2
        \end{array}$ 
    & $\begin{array}{c}
            A_0\oplus A_1\oplus A_2 \\
            B_0\oplus B_1\oplus B_2
        \end{array}$ \\ \hline
        
$Z_{00}$,$Z_{10}$ & $A,A$ & $A,B$ & $B,A$ & $B,B$ \\
$Z_{00}$,$Z_{11}$ & $A,B$ & $A,A$ & $B,B$ & $B,A$ \\
$Z_{01}$,$Z_{11}$ & $B,B$ & $B,A$ & $A,B$ & $A,A$ \\
$Z_{01}$,$Z_{10}$ & $B,A$ & $B,B$ & $A,A$ & $A,B$ \\ \hline
\end{tabular}
\label{tab:schN2K2num2}
}
\end{table}
In the placement phase, the server places either $Z_{i0}$ or $Z_{i1}$, with equal probability, as the cache $Z_i$ for User~$i$. The actual choice is private between the server and User~$i$. 
If User~$i$ was assigned the cache $Z_{i0}$, then 
the multicast transmissions $X^{D_0D_1}$ for each possible demand $(D_0, D_1)$ are as shown in Table \ref{tab:schN2K2num1}. 
It can be seen that all the demands are served. 
It can also be checked that the demands are private under this assignment.
For instance, from Table~\ref{tab:schN2K2num2}, we see that there exists another assignment of cache for each user which recovers another file with the same transmission.

Table \ref{tab:rec1} is the set of recoverable files under each possible cache content for a given transmission.  
\begin{table}[htb]
\centering
\caption{Files recovered from possible caches for a $(2, 2; 1, 2/3)$ private scheme.}
\begin{tabular}{|c|cccc|}
\hline
 & $X^{AB}$ & $X^{BA}$ & $X^{BB}$ & $X^{AA}$ \\ \hline
$Z_{00}$ & $A$ & $B$ & $B$ & $A$ \\
$Z_{01}$ & $B$ & $A$ & $A$ & $B$ \\
$Z_{10}$ & $B$ & $A$ & $B$ & $A$ \\
$Z_{11}$ & $A$ & $B$ & $A$ & $B$ \\ \hline
\end{tabular}
\label{tab:rec1}
\end{table}
For the same transmission, each user is able to recover either file $A$ or file $B$ with the two possible cache contents. Since the actual cache content is private, we readily see that this scheme satisfies the demand privacy condition in Eq.~\eqref{eq:dpcond}.

\subsection{Dual private schemes} 
We show that a $(2,2; M=M_1, R=R_1)$ scheme with demand privacy can be converted into a $(2,2; M=R_1, R=M_1)$ demand-private scheme and this results in symmetric $R$ vs $M$ capacity bounds for the $(2,2)$ case. 

One can observe that the roles of caches and transmissions can be interchanged in the symmetric file recovery matrix in Table \ref{tab:rec1}. Hence, from the scheme given in Table~\ref{tab:schN2K2num1},  we can arrive at a scheme given in Table~\ref{tab:schN2K2dual} with rate $R=1$ for $M=2/3$. We call this scheme the dual of the original scheme.
\begin{table}[htb]
    \centering
    \caption{Dual private $(2, 2; 2/3, 1)$ scheme from the private $(2,2; 1,2/3)$ scheme given in Table~\ref{tab:schN2K2num1}. For cache $Z_{i0}$ at User~$i$, $X^{D_0D_1}$ is the transmission for the demand $D_0D_1$. }
\parbox{.56\linewidth}{
\begin{tabular}{c m{7em}}
        \toprule
        Notation & Possible Cache Contents \\ \midrule
        $Z_{00}$ &  $
        \begin{array}{c}
	        A_1 \\
	        B_1 
        \end{array}
        $\\\midrule[0.2pt]
        $Z_{01}$ &  $
        \begin{array}{c}
	        A_2 \\
	        B_2 
        \end{array}
        $\\\midrule[0.2pt]
        $Z_{10}$ &  $
        \begin{array}{c}
	        A_0\oplus A_1\oplus A_2 \\
	        B_0\oplus B_1\oplus B_2
        \end{array}
        $ \\\midrule[0.2pt]
        $Z_{11}$ & $
        \begin{array}{c}
	        A_0 \\ 
	        B_0
        \end{array}
        $\\ \bottomrule
   \end{tabular}} \quad
\parbox{.38\linewidth}{
    \begin{tabular}{cc}
        \toprule
        $D_1D_2$ & $X^{D_0D_1}$ \\ \midrule
        $AA$ &  $
        \begin{array}{c}
	        A_0\oplus A_2 \\
	        B_0\oplus B_2 \\
	        A_2\oplus B_1
        \end{array}
        $\\\midrule[0.2pt]
        $AB$ &  $
        \begin{array}{c}
	        A_0\oplus A_1 \\
	        B_0\oplus B_1 \\
	        A_2\oplus B_1
        \end{array}
        $\\\midrule[0.2pt]
            $BA$ &  $
        \begin{array}{c}
	        A_0\oplus A_1 \\
	        B_0\oplus B_1 \\
	        A_1\oplus B_2
        \end{array}
        $ \\\midrule[0.2pt]
        $BB$ & $
        \begin{array}{c}
	        A_0\oplus A_2 \\
	        B_0\oplus B_2 \\
	        A_1\oplus B_2
        \end{array}
        $\\ \bottomrule
    \end{tabular} 
   }
    \label{tab:schN2K2dual}
\end{table}

Our next result generalizes the above for all $(2,2)$ private schemes that use one of two caches uniformly at random. 

\begin{lemma}[Duality of transmissions and caches]\label{th:duality}
Suppose that there exists a $(2, 2 ; M=M_1, R=R_1) $ private scheme where the server places one of two possible cache contents uniformly at random. Then, there exists a  $(2, 2; M=R_1, R=M_1) $ private scheme. 
\end{lemma}
\begin{proof}
Consider a $(2, 2 ; M, R) $ private scheme constructed with users having two options to populate their caches. Let $\{Z_{00}, Z_{01}\}$ be the set of two cache options for User~0 and $\{Z_{10}, Z_{11}\}$ be the set of two cache options for User~1. Let $X^{D_1D_2}$ be the transmission corresponding to the user demands $D=(D_1, D_2)$ and cache $Z_{i0}$ at User~$i$. The sets $\mathcal{Z}=\{Z_{00}, Z_{01}, Z_{10}, Z_{11}\}$ and $\mathcal{X}=\{X^{AA}, X^{AB}, X^{BA}, X^{BB}\}$ are able to recover files $A$ and $B$ as given in Table~\ref{tab:rec1}. Let the size of $Z_i$ be $R_1F$ bits and that of $X^{W_0W_1}$ be $M_1F$ bits. We can interchange the role of these caches and transmissions. Let $\{X^{AB}, X^{BA}\}$ be the set of two cache options for User~0 and $\{X^{BB}, X^{AA}\}$ be the set of two cache options for User~1. Then if $Z_{11}$ is transmitted and the Users 0 and 1 are assigned $\{X^{AB}\}$ and $\{X^{BB}\}$ as their caches, both can recover file $A$. Instead if User~1 had $X^{BA}$ in its cache, the users would have recovered $B$ and $A$, respectively. This way of interchangeability between caches and transmissions gives rise to a new scheme for 2 users and 2 files, where the cache size is $M_1$ bits and transmission size is $R_1$ bits.
\end{proof}
A consequence of the above duality is that the achievable trade-off between memory and rate for 
$(2,2)$ private schemes is symmetric about the line  $M=R$. 

\begin{lemma}[Time sharing with file splitting]\label{lm:timesharing} Given two achievable $(M,R)$ pairs for a 
$(2,2)$ private scheme, all values of $(M,R)$ along the line joining these points are achievable.
\end{lemma}
\begin{proof}
Consider $0 \leq \alpha \leq 1$. 
Split the file $A$ into two parts $A_\alpha$ and $A_{\bar{\alpha}}$ of size 
$\alpha F$ bits and $(1-\alpha)F$ bits respectively. 
Similarly, split $B$ into $B_\alpha $ and $B_{\bar{\alpha}}$.
Denote the two achievable private caching schemes as $(2, 2; M, R) $ and $(2, 2; M', R') $
respectively. 
We can use the $(2, 2; M, R)$ scheme for sharing $A_\alpha$ and $B_{\alpha}$ 
 and the $(2, 2; M', R') $  scheme for sharing $A_{\bar{\alpha}}$, and $B_{\bar{\alpha}}$.
 The overall scheme shares $A$ and $B$ with effective cache size $(\alpha M +(1-\alpha) M')F$ bits
 and transmission $(\alpha R+(1-\alpha)R')F$ bits
 giving a $(2, 2; \alpha M+(1-\alpha)M', \alpha R+(1-\alpha)R' ) $ private scheme. 
\end{proof}
Note that the time sharing scheme in Lemma~\ref{lm:timesharing} has a subpacketization that is 
equal to the sum of the two 
schemes used for time sharing.
Using Lemma~\ref{lm:timesharing}, and Lemma~\ref{th:duality} we can plot the upper bounds for the achievable $(M,R)$ pair for $(2, 2)$ private schemes. 
The plot is symmetric about the line $M=R$ as can be seen in Fig.~\ref{fig:mr} 

\begin{figure}
    \centering
    \begin{tikzpicture}
    	\begin{axis}[
    	    width=7cm, height=7cm,
			axis x line=center, 
			axis y line=middle, 
    		xlabel=$M$,ylabel=$R$,
            xmin=0, xmax=2.5,
            ymin=0, ymax=2.5,
        	domain=0:pi/2,
        	xtick={
        		0, 0.66, 1, 2
        	},
        	xticklabels={
        		0, $\frac{2}{3}$, 1, 2
        	},
        	ytick={
        		0.66, 1, 2
        	},
        	yticklabels={
        		$\frac{2}{3}$, 1, 2
        	}
    		]
    
    	\addplot[color=blue,mark=*] coordinates {
    		(0,2)
    		(0.66,1)
    		(1,0.66)
    		(2,0)
    	};
    	\draw [black, thin,dashed] (0,1) -- (0.66,1);
    	\draw [black, thin,dashed] (0,0.66) -- (1,0.66);
    	\draw [black, thin,dashed] (1,0) -- (1,0.66);
    	\draw [black, thin,dashed] (0.66,0) -- (0.66,1);
    	\end{axis}%
    \end{tikzpicture}%
    \caption{Achievable $(M,R) $ region  for $(2, 2; M, R)$ private schemes.
    The $(2, 2; 1, 2/3)$ scheme and its dual scheme  $(2, 2; 1, 2/3)$  have a subpacketization of three subfiles. 
    The straight lines are due to Lemma~\ref{lm:timesharing}. For these schemes, the subpacketization need not be three. }
    \label{fig:mr}
\end{figure}
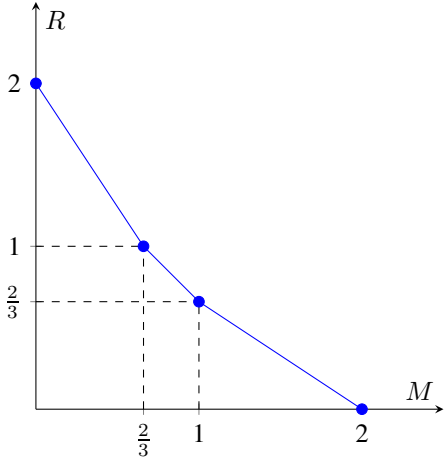

\subsection{Towards lower bounds and optimal subpacketization}
In the non-private case, the $M$ vs $R$ region is fully characterized for two users and two files. For the case with demand privacy, it is not clear whether any of the points in the achievable $M$ vs $R$ curve shown in Fig. \ref{fig:mr} are optimal, or if the subpacketizations are optimal. 

While we do not have lower bounds and optimality results yet, we present a few basic impossibility results involving subpacketization and coding of cache contents. 

In the non-private case for two users/files, subpacketization of 2 suffices to result in optimal rate of $R=1/2$ for $M=1$. For the private case, we have the following result.

\begin{lemma}
Consider $N=2$, $K=2$ with subpacketization of 2 and $M=1$. A rate $R=1/2$ cannot be achieved with demand privacy when using a linear scheme.
\label{lem:nosub2}
\end{lemma}
\begin{proof}
A proof is given in Appendix~\ref{sec:nogo}.
\end{proof}

For subpacketization of 3, the scheme in Section \ref{sec:sub3} uses coded cache contents, which is not typical in the non-private setting. In the setting considered here for demand privacy, we have the following result on coding in cache contents.

\begin{lemma}
Consider $N=2$, $K=2$ with subpacketization of 3 and $M=1$. If the cache contents are not allowed to be coded (i.e. linear combinations of two or more file parts are not allowed to be stored in cache), a rate $R=2/3$ cannot be achieved with demand privacy when using a linear scheme.
\label{lem:nosub3noncoded}
\end{lemma}
\begin{proof}
A proof is given in Appendix~\ref{sec:no-uncoded}.
\end{proof}

\section{General Scheme and Partial Privacy}\label{sec:gen-scheme}
In this section, we describe the general scheme from \cite{kamath2019demand} that provides the design of a demand-private coded caching scheme from non-private schemes.

\begin{theorem}[Existence of private schemes \cite{kamath2019demand}]\label{th:existence}
If there exists  a $(KN,N;M,R)$ coded caching scheme, then there exists a private $(K,N;M,R)$ scheme.
\end{theorem}
\begin{proof}  

Assume that we have a $(KN, N; M, R)$ non-private scheme. 
Let the cache contents of each of the users be given as 
$Z_i'$, where $0\leq i <NK$.

Partition the users into sets of size $N$.
Without loss of generality we partition the $NK$ users as 
\begin{eqnarray}
\mathcal{U}_k= \{(k-1)N \leq j <k N \}. \label{eq:users-in-k}
\end{eqnarray}

Denote the cache of the $k$th user of the private scheme as $Z_k$.
This is chosen as follows:
\begin{eqnarray}
{Z}_k &=& Z_{(k-1)N+r_k}' \label{eq:cache-zk-npsch}
\end{eqnarray}
where $r_k$ is uniformly distributed 
on $\{0, 1, \ldots, N-1 \}$.

During delivery the server receives the demand vector $(d_0,\ldots, d_{K-1})$.
The server then generates the transmission corresponding to the demand vector of the
non-private scheme. 
This demand vector is of length $NK$ and denoted $D'=(d_j')$.
We can assign any random permutation of the demands $[N] $ to the users in
$\mathcal{U}_k$ subject to the condition that 
 $d_{r_k+(k-1)N}' = d_k$. 
Formally, 
\begin{subequations}
\begin{eqnarray}
\pi_k&:&\mathcal{U}_k   \rightarrow  [N] \label{eq:pi-k}\\
d_j'&=& \pi_k(j)\label{eq:pi-k-1}\\
d_{r_k+(k-1)N}'&=& \pi_k(r_k+(k-1)N) =d_k\label{eq:pi-k-2}
\end{eqnarray}
\end{subequations}

Denote the demand vector of the non-private
$(KN, N)$ scheme as 
$D'=(d_j')_{j\in [NK]}$
Since the non-private scheme can accommodate all demands, it can also serve this demand.
Transmit $X^{D'}$ as per the non-private scheme. 
Then each user of the private scheme is able to receive the file requested.

Demand privacy can be shown as follows. 
The $i$-th user of the  private scheme is able to recover the file he or she requested. 
The same transmission can be used to recover all the files by the caches 
$Z_{(j-1)K}', \ldots, Z_{jK-1}'$.
However, the $i$-th user does not know which of these caches has been assigned to the $j$-th user. 
Since all of them are equally likely to be assigned to $j$-th user by construction, the uncertainty about the demand $D_j$ given $D_i, X, Z_i$ is $H(D_j)$.
Thus, the privacy of demands is preserved. 

Observe that the cache size of users in the private scheme is same as the size of the cache in the non-private scheme. 
Similarly, the rate of transmission for the 
private scheme is exactly the same as that of the non-private scheme. 
From this it follows the demand private scheme has the parameters $(K, N; M, R)$ as claimed. 
\end{proof}

\begin{remark}[Extended demand vector]
While creating the extended demand vector $D'$ we can make a simple choice for
$\pi_k$.
The demand of the $j$th user of the non-private scheme 
is given as 
\begin{eqnarray}
d_j' = d_k-r_k+j \bmod N \mbox{ for } (k-1)N\leq j<kN,\label{eq:dd-np-KN-N}
\end{eqnarray} 
where $0\leq k <K$.
\end{remark}

\subsection{Constructions using Maddah-Ali-Niesen schemes and PDAs}
 Using the Maddah-Ali-Niesen scheme \cite{maddah2014fundamental} as the non-private scheme in Theorem \ref{th:existence}, we obtain the following:
 \begin{corollary}
 There exists a demand private $(K, N; M, R)$ scheme for integer values of $KM$, where the rate 
   \begin{align}
      R = 
  \begin{cases} 
   \frac{K(N-M)}{(1+KM)} & \text{if } M \geq \frac{K-1}{K} \\
   N-M       & \text{if } M < \frac{K-1}{K}
  \end{cases}.
  \end{align}
 \end{corollary}
 \begin{proof}
 This follows from Theorem~\ref{th:existence} using the scheme proposed by Maddah Ali and Niesen \cite[Theorem~1]{maddah2014fundamental}.
 In this case, for integer values of $KM$ we can construct a $(NK, N; M , R)$ non-private scheme.
 If $KM\geq K-1$, then 
  $R=\frac{K(N-M)}{(1+KM)}$. 
  If $1\leq KM < K-1$, then $R=N-M$.
  We can map each user to a user in the non-private scheme using Eq.~\eqref{eq:cache-zk-npsch} and extend the demand vector of the private $(K,N;M,R)$ scheme to the non-private scheme using Eq.~\eqref{eq:dd-np-KN-N}. Then the scheme from \cite{maddah2014fundamental} gives the cache contents that should be stored in each user and a transmission for each demand from which each user can recover their files. The cache memory and rate required in the private scheme will be the same as that in the non-private scheme.
 \end{proof}
 Note that there is no coding gain when $KM < K-1$.
 
A general framework for non-private coded caching schemes was proposed in \cite{yan2017placement} using placement delivery arrays (PDAs).
We can convert many of these schemes to private coded caching schemes. 
Some of them improve upon those derived from schemes \cite{maddah2014fundamental} in subpacketization or other parameters. For positive integers $K$, $f$, $Z$ and $S$, a $(K, f, Z, S)$ placement delivery array is a $f\times K $ matrix  $(P=[p_{j,k}]\mbox{ with }j\in[F],k\in[K])$ containing either a ``$\Asterisk$'' or integers from $\{0,1,\ldots,S-1\}$ in each cell such that they satisfy a few conditions \cite{yan2017placement}. 
Here, $f$ is the subpacketization, and $S$ is the total number of transmissions each of size $1/f$ of the file.
For any $N$, we can obtain a $(K,N;\frac{NZ}{f},\frac{S}{f})$ coded caching scheme from a $(K, f, Z, S)$ placement delivery array.

\begin{corollary}[Private schemes from PDAs]\label{co:pda}
If there exists a $(NK, f, Z, S)$ placement delivery array, we can obtain a private $(K, N; \frac{NZ}{f},\frac{S}{f})$ coded caching scheme, for any $N$.
\end{corollary}
\begin{proof}
Given a $(NK, f, Z, S)$ placement delivery array, there exists a non-private $(NK,N;\frac{NZ}{f},\frac{S}{f})$ (see \cite{yan2017placement} for details). From this we can obtain the private $(K,N;\frac{NZ}{f},\frac{S}{f})$ scheme using Theorem~\ref{th:existence}.
\end{proof}

We now present an example of a private scheme with $N=2$, $K=3$, derived from a PDA.
Consider the PDA from  \cite[Eq.~(7)]{yan2017placement} corresponding to 6 users and 4 subfiles.
\begin{align}
    P=\left[\begin{array}{llllll}{*} & {1} & {*} & {2} & {*} & {0} \\ {0} & {*} & {*} & {3} & {1} & {*} \\ {*} & {3} & {0} & {*} & {2} & {*} \\ {2} & {*} & {1} & {*} & {*} & {3}\end{array}\right]\label{eq:pda}
\end{align}
We assume that each file $W_i$ is split into 
$f$ subfiles which are denoted as $W_{i,j}$, where 
$0\leq j<f$.
In the non-private scheme, the cache contents of the $i$-th  user are given below. 

\begin{align*}
    Z_{0}' &=\left\{W_{i, 0}, W_{i, 2} : i \in[0,6)\right\} \\
    Z_{1}' &=\left\{W_{i, 1}, W_{i, 3} : i \in[0,6)\right\} \\
    Z_{2}' &=\left\{W_{i, 0}, W_{i, 1} : i \in[0,6)\right\} \\
    Z_{3}' &=\left\{W_{i, 2}, W_{i, 3} : i \in[0,6)\right\} \\
    Z_{4}' &=\left\{W_{i, 0}, W_{i, 3} : i \in[0,6)\right\} \\
    Z_{5}' &=\left\{W_{i, 1}, W_{i, 2} : i \in[0,6)\right\} 
\end{align*}
The transmission for demand vector $\boldsymbol{d}'=(d_0', \ldots , d_5')$ is 
\begin{eqnarray}
X^{d'}=\left \{\begin{array}{c}
W_{d'_0,1} \oplus W_{d'_2,2} \oplus W_{d'_5,0}\\ 
W_{d'_1,0} \oplus W_{d'_2,3} \oplus W_{d'_4,1}\\ 
W_{d'_0,3} \oplus W_{d'_3,0} \oplus W_{d'_4,2}\\ 
W_{d'_1,2} \oplus W_{d'_3,1} \oplus W_{d'_5,3}
\end{array}
\right\}.\label{eq:dp-26-pda-X-d}
\end{eqnarray}

For $N=2$ files, $A$ and $B$, we can create a private (3, 2; 1, 1) scheme as shown in Fig.~\ref{fig:schN2K3}.
 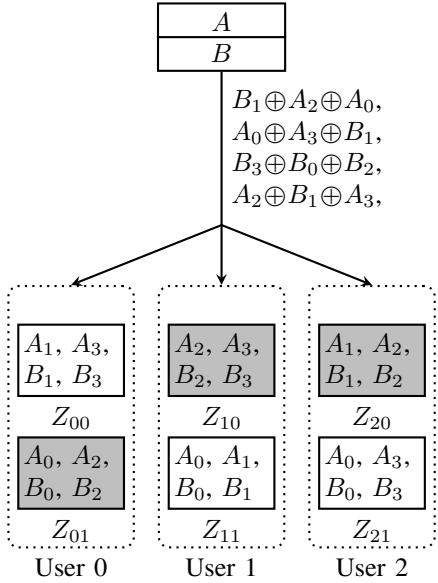
\begin{figure}[htb]
     \centering
     \begin{tikzpicture}[>=stealth, thick,
     every fit/.style={
		rounded corners,
		draw,
		inner ysep=0.5cm
	}
     ]
     \tikzstyle{block} = [rectangle, draw,thick,fill=blue!0,
    text centered, rounded corners, minimum height=1em]
        \node[draw, thick, align=center, inner xsep=20pt, minimum height=6mm] (s) {$A$\\$B$};
        \coordinate [below of=s,node distance=2.5cm] (b1) {};
        \node[draw, thick, align=center, fill=gray!50, inner xsep=2pt, minimum height=6mm, below of=b1,label=below:$Z_{10}$,node distance=1.8cm] (z2) {%
        \begin{varwidth}{4em}
            $A_2$, $A_3$, $B_2$, $B_3$
        \end{varwidth}
        };
        \node[draw, thick, align=center, inner xsep=2pt, minimum height=6mm, left of=z2,label=below:$Z_{00}$,node distance=2cm] (z0) {%
        \begin{varwidth}{4em}
            $A_1$, $A_3$, $B_1$, $B_3$
        \end{varwidth}
        };
        \node[draw, thick, align=center, fill=gray!50, inner xsep=2pt, minimum height=6mm, below of=z0,label=below:$Z_{01}$,node distance=1.5cm] (z1) {%
        \begin{varwidth}{4em}
            $A_0$, $A_2$, $B_0$, $B_2$
        \end{varwidth}
        };
        \node[draw, thick, align=center, inner xsep=2pt, minimum height=6mm, below of=z2,label=below:$Z_{11}$,node distance=1.5cm] (z3) {%
        \begin{varwidth}{4em}
            $A_0$, $A_1$, $B_0$, $B_1$
        \end{varwidth}
        };
        \node[draw, thick, align=center, fill=gray!50, inner xsep=2pt, minimum height=6mm, right of=z2,label=below:$Z_{20}$,node distance=2cm] (z4) {%
        \begin{varwidth}{4em}
            $A_1$, $A_2$, $B_1$, $B_2$
        \end{varwidth}
        };
        \node[draw, thick, align=center, inner xsep=2pt, minimum height=6mm, below of=z4,label=below:$Z_{21}$,node distance=1.5cm] (z5) {%
        \begin{varwidth}{4em}
            $A_0$, $A_3$, $B_0$, $B_3$
        \end{varwidth}
        };
        \path[thick] (s) edge node[text width=2cm, right] {
            $B_1 \oplus A_2 \oplus A_0$,
            $A_0 \oplus A_3 \oplus B_1$,
            $B_3 \oplus B_0 \oplus B_2$,
            $A_2 \oplus B_1 \oplus A_3$,
        } (b1);
    	\node[draw,dotted,fit=(z0) (z1),minimum width=1.5cm, label=below:User~0] (u1){} ;
    	\node[draw,dotted,fit=(z2) (z3),minimum width=1.5cm, label=below:User~1] (u2) {} ;
    	\node[draw,dotted,fit=(z4) (z5),minimum width=1.5cm, label=below:User 2] (u3) {} ;
        \draw[->,thick] (b1) -- (u1.north);
        \draw[->,thick] (b1) -- (u2.north);
        \draw[->,thick] (b1) -- (u3.north);
        \draw[-,thick] (s.west) -- (s.east);
     \end{tikzpicture}
     \caption{A $(3,2; 1, 1)$ private scheme for $D=(A, A, B)$ from a $(6,2; 1, 1) $ non-private scheme from the PDA given in \eqref{eq:pda}.}
     \label{fig:schN2K3}
 \end{figure}
 
\subsection{Case of two files, two users}
 For the $N=2$, $K=2$ case considered earlier, the $M=1$, $R=2/3$ construction presented in Section \ref{sec:sub3} is not derived from a non-private scheme but constructed directly. In fact, a construction from the Maddah-Ali-Niesen scheme using Theorem \ref{th:existence} results in a subpacketization of 6, when compared to the subpacketization of 3 needed for the scheme in Section \ref{sec:sub3}. This shows that direct construction has the benefits of improved subpacketization.
 
 \subsection{Partial privacy and reduction in subpacketization}
 The scheme modified from the non-private scheme can have less subpacketization if full privacy is not needed. For instance, suppose that 2-file privacy suffices. That is, at the end of the multicast transmission, every user has an ambiguity of one of two files about any other user's demand. 
 
For 2-file privacy, we need to provide only two options to populate the cache content of a user.  Hence, we can use a $(2N,K)$ non-private scheme to arrive at an $(N,K)$ partially private scheme where any user's demand is possibly one of two files to another user. These schemes are important particularly when we have large number of files compared to users. For example, if $N=10$ and $K=2$, then a fully private scheme modified from the non-private scheme would require the non-private scheme to have $K'=NK = 20$. With $M=5$, such a scheme would require a subpacketization $f = \binom{K'}{\frac{K'M}{N}} = \binom{20}{10} = 184756$. But under 2-file privacy for this setup, $K'=2 K=4$, and we can use a subpacketization as low as $f = \binom{4}{2} = 6$. In Fig.~\ref{fig:schN4K4partial}, we show a partially private $(2,4; 2, 2/3)$ scheme from a $(4,4; 2, 2/3) $ non-private scheme providing an ambiguity of two files.

 \begin{figure}[htb]
     \centering
     \begin{tikzpicture}[>=stealth, thick,
     every fit/.style={
		rounded corners,
		draw,
		inner ysep=0.5cm
	}
     ]
        \node (s) [rectangle split, rectangle split parts=4, minimum width=2cm,  draw, anchor=center] {A\nodepart{second}B\nodepart{third}C\nodepart{fourth}D} ;
        \coordinate [below of=s,node distance=3cm] (b1) {};
        \coordinate [below of=b1,node distance=2.2cm] (b2) {};
        \node[draw, thick, align=center, fill=gray!50, inner xsep=2pt, minimum height=6mm, left of=b2,label=below:$Z_{01}$,node distance=1.1cm] (z1) {%
        \begin{varwidth}{5em}
            $A_{01}$, $A_{12}$, $A_{13}$, $B_{01}$, $B_{12}$, $B_{13}$,
            $C_{01}$, $C_{12}$, $C_{13}$, $D_{01}$, $D_{12}$, $D_{13}$
        \end{varwidth}
        };
        \node[draw, thick, align=center, inner xsep=2pt, minimum height=6mm, left of=z1,label=below:$Z_{00}$,node distance=1.8cm] (z0) {%
        \begin{varwidth}{5em}
            $A_{01}$, $A_{02}$, $A_{03}$, $B_{01}$, $B_{02}$, $B_{03}$,
            $C_{01}$, $C_{02}$, $C_{03}$, $D_{01}$, $D_{02}$, $D_{03}$
        \end{varwidth}
        };
        \node[draw, thick, align=center, fill=gray!50, inner xsep=2pt, minimum height=6mm, right of=b2,label=below:$Z_{10}$,node distance=1.1cm] (z2) {%
        \begin{varwidth}{5em}
            $A_{02}$, $A_{12}$, $A_{23}$, $B_{02}$, $B_{12}$, $B_{23}$,
            $C_{02}$, $C_{12}$, $C_{23}$, $D_{02}$, $D_{12}$, $D_{23}$
        \end{varwidth}
        };
        \node[draw, thick, align=center, inner xsep=2pt, minimum height=6mm, right of=z2,label=below:$Z_{11}$,node distance=1.8cm] (z3) {%
        \begin{varwidth}{5em}
            $A_{03}$, $A_{13}$, $A_{23}$, $B_{03}$, $B_{13}$, $B_{23}$,
            $C_{03}$, $C_{13}$, $C_{23}$, $D_{03}$, $D_{13}$, $D_{23}$
        \end{varwidth}
        };
        \path[thick] (s) edge node[text width=3cm, right] {
        $D_{12} \oplus B_{02} \oplus D_{01}$, 
        $D_{13} \oplus B_{03} \oplus C_{01}$, 
        $D_{23} \oplus D_{03} \oplus C_{02}$, 
        $B_{23} \oplus D_{13} \oplus C_{12}$ 
        } (b1);
    	\node[draw,dotted,fit=(z0) (z1),minimum width=1.5cm, label=below:User~0] (u1){} ;
    	\node[draw,dotted,fit=(z2) (z3),minimum width=1.5cm, label=below:User~1] (u2) {} ;
        \draw[->,thick] (b1) -- (u1.north);
        \draw[->,thick] (b1) -- (u2.north);
        \draw[-,thick] (s.west) -- (s.east);
     \end{tikzpicture}
     \caption{A $(2,4; 2, 2/3)$ partially private scheme from a $(4,4; 2, 2/3) $ non-private scheme. The scheme has a privacy of two files. The gray boxes show the cache assigned by the server. The demands corresponding to unassigned caches for User $k$ are selected at random from the set $\{W_{j,j\in[N]}\}\setminus D_k$. Transmission shown is for the demand vector $D=\{B, D\}$ and the extended demand vector $D=\{D, B, D, C\}$. Observe that cache contents $Z_{00}$, $Z_{01}$, $Z_{10}$, $Z_{11}$ recover the files 
     $D$, $B$, $D$, $C$, respectively. From the point of view of the User~1, the User~0 could have requested either $D$ or $B$ giving the necessary privacy.}
     \label{fig:schN4K4partial}
 \end{figure}
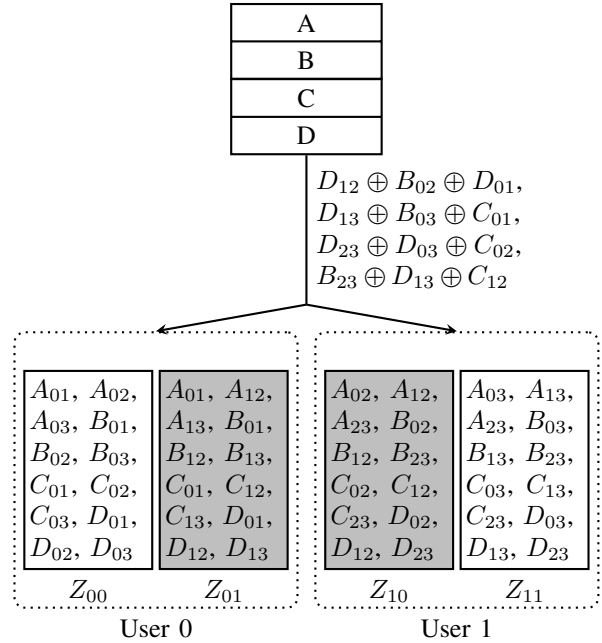

\section{Conclusion}\label{sec:conc}
We have investigated here the problem of demand privacy in systems employing coded caching techniques with a focus on minimizing subpacketization. For the 2-user, 2-file case, we provided a new construction with a subpacketization of 3. Additionally, we proved that the subpacketization of $3$ is indeed minimal for a linear code for the 2-user, 2-file case. Also, we proposed  partially private caching schemes and showed how to construct  such private schemes with less subpacketization in the general $K$-user, $N$-file case.

\appendices
\section{Impossibility results for $(2, 2) $ private linear  schemes with two subfiles}\label{sec:nogo}
A coded caching scheme is said to be linear if all the cache contents, transmissions
and the decoding involves only linear operations.
Here we provide a proof for Lemma~\ref{lem:nosub2} in Section~\ref{sec:22}  and show that there does not exist a private $
(2, 2; 1, 0.5 )$ linear coded caching scheme with subpacketization of two. The proof method is by contradiction. So, we begin by assuming the existence of a $(2,2;1,0.5)$ linear coded caching scheme with subpacketization of two.

\subsection{Notation and setup}
Suppose $A,  B$ are split into two subfiles each as $A_0, A_1$ and $B_0, B_1$, respectively. 
Let $S$ be defined as 
\begin{eqnarray} 
S =\left[\begin{array}{c} A_0\\ A_1\\ B_0\\ B_1 \end{array}\right].\label{eq:subfiles}
\end{eqnarray}
Let the $i$-th user's cache $Z_i$ and the transmission $X^{D_1D_2}$ be written as 
\begin{eqnarray}
Z_i &=& \mathsf{C}_i S, \label{eq:coeff-matrix}\\
X^{D_1D_2} &=& T^{D_1D_2} S, \label{eq:tx-matrix}
\end{eqnarray}
where $\mathsf{C}_i$ and $T^{D_1D_2}$ are $2\times 4$ and $1\times 4$ coefficient matrices, respectively, with entries from a suitable field. User~$i$ can decode $D_i$ given $T^{D_1D_2}$ using the cache $Z_i$.

The matrix $\mathsf{C}_i$ is split into $2\times 2$ matrices $\mathsf{C}_{iA}$ and $\mathsf{C}_{iB}$ as follows:
\begin{eqnarray}
\mathsf{C}_i& =& \left[ \begin{array}{cc} \mathsf{C}_{iA}& \mathsf{C}_{iB}\end{array} \right],
\end{eqnarray}
Similarly, $T^{D_1D_2}$ is split into two $1\times 2$ submatrices as shown below. 
\begin{eqnarray}
T^{D_1D_2} &=& \left[ \begin{array}{cc} T^{D_1D_2}_{A}& T^{D_1D_2}_{B}\end{array} \right].
\end{eqnarray}
We denote the rank of a matrix $ Q $ by $ \rk(Q) $.

We assume that 
	\begin{equation}
	\rk(\mathsf{C}_i) =2 
	\end{equation}
implying that each cache contains independent subfile combinations.

\subsection{Lemmas on structure of coefficient matrices}
\begin{lemma}[Rank constraints on coefficient matrices] \label{lm:ci-rank}
$ \rk(\mathsf{C}_{iA})= \rk(\mathsf{C}_{iB}) =1$.
\end{lemma}
\begin{proof}
We will show $\rk(\mathsf{C}_{0B})=1$. Consider the $3\times 4$ matrix 
\begin{equation}
    M_{0AB}=\begin{bmatrix}
    \mathsf{C}_{0A}&\mathsf{C}_{0B}\\
    T^{AB}_A   &  T^{AB}_B
    \end{bmatrix}.
\end{equation}
The cache of User~0 and transmission for the demand $AB$, when combined, result in the vector $M_{0AB} S$. Since User~0 can recover $A = [A_0\; A_1]$ by linearly combining the elements of $M_{0AB} S$, there exists a $2\times 3$ matrix $U$ such that 
\begin{equation}
    U\begin{bmatrix}
    \mathsf{C}_{0A}\\
    T^{AB}_A
    \end{bmatrix} = \begin{bmatrix}
    1&0\\
    0&1
    \end{bmatrix},
    \label{eq:U1}
\end{equation}
and
\begin{equation}
    U\begin{bmatrix}
    \mathsf{C}_{0B}\\
    T^{AB}_B
    \end{bmatrix} = \begin{bmatrix}
    0&0\\
    0&0
    \end{bmatrix},
    \label{eq:U2}
\end{equation}
resulting in the decoding of $A$ and elimination of $B_0$ and $B_1$. From Eq.~\eqref{eq:U1}, the rank of $U$ is 2. Using this in Eq.~\eqref{eq:U2}, $\rk(\mathsf{C}_{0B})\ne 2$. 

If $\mathsf{C}_{0B}$ is the all-zero matrix, then User~0 cannot recover $B$ using only the transmission $X^{BA}$. So, $\rk(\mathsf{C}_{0B})\ne 0$. 
This only leaves the possibility $\rk(\mathsf{C}_{0B})=1$. The proof above can be readily adapted to show $\rk(\mathsf{C}_{iA})=1$ for $i=0,1$ and $\rk(\mathsf{C}_{1B})=1$.
\end{proof}

A consequence of Lemma~\ref{lm:ci-rank}, none of the files are stored entirely on any cache.

For an invertible $2\times 2$ matrix $U$, the scheme obtained by replacing $\mathsf{C}_i$ by $U \mathsf{C}_i$ is also a demand-private coded caching scheme because one cache can be obtained from the other. This is captured in the following lemma for future use.
\begin{lemma}[Equivalent coefficient matrices]\label{lm:equivalent-Ci}
Cache $Z_i=C_iS$ can recover file $W$ from a transmission $X$ iff $Z_i'=U C_iS$ can recover the same file from $X$ for any invertible $2\times 2$ matrix $U$.
\end{lemma}

\begin{corollary}[Reduced coefficient matrices]\label{co:cor-pure}
	Given the coefficient matrix $\mathsf{C}_i$,	
	there exist invertible matrices $U_i$ and $V_i$ such that 
	    \begin{eqnarray}
	V_iU_i \mathsf{C}_i = \left[\begin{array}{cccc} a&b&0&0 \\
	 0&0&c&d\end{array} \right],\label{eq:reduced-Ci}
	 	\end{eqnarray}
where both $(a,b)$ and $(c,d)$ are nonzero.
\end{corollary}
\begin{proof}
Suppose that $\mathsf{C}_i$ is written as follows.
	\begin{eqnarray}
	\mathsf{C}_i = \left[\begin{array}{cccc} a&b&*&* \\
	 a'&b'&*&*\end{array} \right].
	\end{eqnarray}
By Lemma~\ref{lm:ci-rank}, $\rk(\mathsf{C}_{iA})=1$. 
Hence, without loss of generality we can assume that  $(a, b)\neq (0,0)$, and $(a', b')$
is a scalar multiple of $(a, b)$.
    There exists some invertible matrix $U_i=\left[\begin{array}{cc}1 & 0 \\\alpha&1\end{array}\right]$ for some scalar $\alpha$ such that
    \begin{eqnarray}
	U_i\mathsf{C}_i = \left[\begin{array}{cccc} a&b&c'&d' \\
	 0&0&c&d\end{array} \right]
	\end{eqnarray}
	for some $c$, $d$, $c'$ and $d'$.
	Since $\rk(\mathsf{C}_i)=2$, it follows that $(c, d) \neq (0,0)$.
	Then for some $\beta$ and $V_i=\left[\begin{array}{cc}1 & \beta \\0&1\end{array}\right]$
	 we obtain Eq.~\eqref{eq:reduced-Ci}.
\end{proof}

An immediate consequence of Lemma~\ref{lm:equivalent-Ci} and Corollary~\ref{co:cor-pure}, we can
assume that the coefficient matrices are of the form 
given below.
\begin{eqnarray}
 \mathsf{C}_i=\left[\begin{array}{cccc} a_i&b_i&0&0 \\
	 0&0&c_i&d_i 
	 \end{array}\right].\label{eq:ci-std-form}
	 \end{eqnarray}
	 
\begin{lemma}[Constraints due to recovery]\label{lm:ci-t-rank}
	Given $ T^{D_0 D_1} $ and $\mathsf{C}_i$, we have the following constraints.
	\begin{subequations}   
	\begin{eqnarray}
	\rk\left(\begin{bmatrix}
			\mathsf{C}_{iD_i}\\[3pt]
			T_{D_i}^{D_0D_1}
			\end{bmatrix}\right)=2 \label{eq:ct-rank-2}\\
			\rk\left(\begin{bmatrix}
			\mathsf{C}_{i\overline{D_{i}}} \\[3pt]
			T_{\overline{D_i}}^{D_0D_1}
			\end{bmatrix}\right)= 1,\label{eq:ct-rank-1}
	\end{eqnarray}
	where $D_i \in \{ A, B\}$ and $\overline{D_i}=\{A,B\}\setminus D_i$ is the file that is not demanded by User~$i$.
	\end{subequations}
	
	\end{lemma}
	\begin{proof}

	    Consider the coefficient matrix $\mathsf{C}_i$ of User~$i$, has an equivalent form given in Eq.~\eqref{eq:ci-std-form}. 
	    Combining with $T^{D_0D_1}$
	 we have 
	 \begin{eqnarray}
	 \left[\begin{array}{c} \mathsf{C}_i \\ T^{D_0D_1} \end{array}\right]S =
	 \left[\begin{array}{cccc} a_i&b_i&0&0 \\
	 0&0&c_i&d_i \\
	 \multicolumn{2}{c}{T^{D_0D_1}_A}&\multicolumn{2}{c}{T^{D_0D_1}_B}\end{array}\right]\left[\begin{array}{c} A_0\\A_1\\B_0\\B_1\end{array}\right].\label{eq:cts-detailed}
	 \end{eqnarray}
	  From this system of equations, one can observe that $A_0$ and $A_1$ appear in two equations. 
	 If $D_i=A$, then User~$i$ must recover 
	  the subfiles $A_0$ and $A_1$, and the following condition must hold.
	 \begin{eqnarray}
	 \rk\left(\left[\begin{array}{cc} a_i&b_i \\
	 \multicolumn{2}{c}{T^{D_0D_1}_A}\end{array}\right]\right)=2\label{eq:ct-invertible1}
	 \end{eqnarray}
	 Similarly if $D_i=B$, the following condition must hold for User~$i$ to recover file $B$ from $T^{D_0D_1}$.
	 \begin{eqnarray}
	 \rk\left(\left[\begin{array}{cc} c_i&d_i \\
	 \multicolumn{2}{c}{T^{D_0D_1}_B}\end{array}\right]\right)=2.\label{eq:ct-invertible2}
	 \end{eqnarray}
	  Eq.~\eqref{eq:ct-rank-2} follows from Eq.~\eqref{eq:ct-invertible1} and Eq.~\eqref{eq:ct-invertible2}. 

One can see that the condition in Eq.~\eqref{eq:ct-invertible1} is not enough for recovering $D_i=A$ 
at User~$i$ using Eq.~\eqref{eq:cts-detailed}. We should be able to remove the part corresponding to $T^{D_0D_1}_{\overline{D_i}}= T^{D_0D_1}_B$ from the third row in Eq.~\eqref{eq:cts-detailed} to arrive at two equations in two variables $A_0$, $A_1$ and solve for them. So if $T^{D_0D_1}_B$ is nonzero, then it should be a scalar multiple of $(c_i,d_i)$. Since $(c_i,d_i)$ is nonzero from Lemma~\ref{lm:ci-rank}, we have
	  \begin{eqnarray}
	 \rk\left(\left[\begin{array}{cc} c_i&d_i \\
	 \multicolumn{2}{c}{T^{D_0D_1}_B}\end{array}\right]\right)=1
	 \end{eqnarray}
	 Hence
 	\begin{eqnarray}
	    \rk\left(\begin{bmatrix}
			\mathsf{C}_{0B} \\
			T_{B}^{AD_1}
		\end{bmatrix}\right)= 1 \mbox{ and }
		\rk\left(\begin{bmatrix}
			\mathsf{C}_{1B} \\
			T_{B}^{D_0A}
		\end{bmatrix}\right)= 1\label{eq:ct-rank-1a}.
	\end{eqnarray}
Similarly solving for $B_0$ and $B_1$ (i.e. $D_i=B$) at User~$i$ requires
	  \begin{eqnarray}
	 \rk\left(\left[\begin{array}{cc} a_i&b_i \\
	 \multicolumn{2}{c}{T^{D_0D_1}_A}\end{array}\right]\right)=1
	 \end{eqnarray}
	 Hence
 	\begin{eqnarray}
	    \rk\left(\begin{bmatrix}
			\mathsf{C}_{0A} \\
			T_{A}^{BD_1}
		\end{bmatrix}\right)= 1 \mbox{ and }
		\rk\left(\begin{bmatrix}
			\mathsf{C}_{1A} \\
			T_{A}^{D_0B}
		\end{bmatrix}\right)= 1\label{eq:ct-rank-1b}.
	\end{eqnarray}
	Eq.~\eqref{eq:ct-rank-1a} and Eq.~\eqref{eq:ct-rank-1b} immediately imply Eq.~\eqref{eq:ct-rank-1}.
	\end{proof}
	
So far, we have not used the requirement of demand privacy. The following lemma uses the demand privacy condition to derive an important constraint on the transmission.
\begin{lemma}[Constraints on transmission]
	\label{lm:tr}
If $X^{D_0D_1}=T^{D_0D_1}S$ where $S$ is defined as in Eq.~\eqref{eq:subfiles}, then $T^{AA}_A$ and $T^{AA}_B$ are both nonzero.
\end{lemma}
\begin{proof}

	 If $T_A^{AA}$ is zero, then User~0 cannot recover file $A$. So, $T^{AA}_A$ is nonzero. 
    We know that the entire file $B$ is not stored on any cache. 
    If $T_B^{AA}$ is zero, then every user must be demanding only $A$.
    This reveals the demands of all the users, 
    so $T_{B}^{AA}$ must be nonzero.

\end{proof}
Note that Lemma~\ref{lm:tr} is only a necessary condition for demand privacy.

\subsection{Proof of Lemma~\ref{lem:nosub2}}
	Let the coefficient matrix $T^{AA}=(u,v,w,x)$. From Lemma \ref{lm:tr}, $(u, v)\neq 0 $ and $(w, x) \neq 0$.
	\begin{eqnarray}
	 \left[\begin{array}{c} \mathsf{C}_i \\ T^{AA} \end{array}\right]S=\left[\begin{array}{cccc} a_i&b_i&0&0 \\
	 0&0&c_i&d_i \\
	 u&v&w&x\end{array}\right]\left[\begin{array}{c} A_0 \\
	 A_1 \\
	 B_0 \\
	 B_1\end{array}\right]
	 \end{eqnarray}
	By Eq.~\eqref{eq:ct-rank-1}, we have
	 \begin{align*}
	     \rk\left(\left[
	     \begin{array}{cc} c_i&d_i \\
			w&x \end{array}
	     \right]\right)=1
	 \end{align*}

	Both  $(c_i,d_i)$ and $(w,x)$ are nonzero due to Lemma~\ref{lm:ci-rank} and Lemma~\ref{lm:tr} respectively. Thus, both $(w,x)$ and $(c_i,d_i)$ are scalar multiples of each other for $i=0,1$. 
	This implies that $(c_0,d_0)$ is a scalar multiple of $(c_1,d_1)$. 
	Then $\rk\left(\left[\begin{array}{c c}
	c_i&d_i\\
	\multicolumn{2}{c}{T^{AB}_B}
	\end{array}\right]\right)$ is same for $i=0, 1$. 
	However, this contradicts  Lemma \ref{lm:ci-t-rank}, by which 	$$
	\rk\left(\left[\begin{array}{c c}
	c_0&d_0\\
	\multicolumn{2}{c}{T^{AB}_B}
	\end{array}\right]\right) = 1 \text{ and }	\rk\left(\left[\begin{array}{c c}
	c_1&d_1\\
	\multicolumn{2}{c}{T^{AB}_B}
	\end{array}\right]\right) = 2.
	$$
	This shows the in-feasibility of coefficient matrices satisfying the rank constraints due to recovery and demand privacy.

	Therefore, we conclude that a linear private $(2,2;1,1/2)$ coded caching scheme with subpacketization of two does not exist.

\section{Impossibility results for uncoded prefetching with three subfiles}\label{sec:no-uncoded}
In Appendix~\ref{sec:nogo}, we have seen that with two subfiles we cannot obtain a private $(2,2;1,1/2)$ scheme. Here we show that without coded prefetching we cannot obtain a private $(2,2;1,2/3)$ scheme with three subfiles and thereby prove Lemma~\ref{lem:nosub3noncoded}.

Informally, the proof is organized as follows. 
First, we show that 
without coded prefetching the subfiles must be cached in an uncoded form i.e. 
without linear combinations. 
This restricts the possibilities for the caches.
Furthermore any given cache restricts the possibilities for the other user's cache. 
Demand privacy is possible only if the set of caches consistent with a user 
allow the reconstruction of both the files for any demand. 
We show that is not possible and hence a linear private $(2,2;1,2/3)$ scheme with subpacketization of three subfiles does not exist.

\subsection{Permissible caches without coded prefetching}
Without coded prefetching, the subfiles can only be replicated in the cache.
With three subfiles, $M=1$ implies that each user can store 3  subfiles. $R=2/3$ implies that there are two independent subfile combinations in the transmission. If all the subfiles in a cache belongs to a file, that user cannot recover the other file from a transmission of rate $R=2/3$. So a cache should contain two subfiles of one file and one subfile of the other file.  Let the two files be A and B. Without loss of generality, let us assume the cache of first user, $Z_0$ contains two subfiles of file A and one subfile of B. 
\begin{eqnarray}
 Z_0=\{A_0, A_1, B_2\}.
\end{eqnarray}
Let the cache of User~1 be
\begin{eqnarray}
Z_1=\{G_0, G_1, G_2\},
\end{eqnarray}
where $G_i\in \{A_0, A_{1}, A_{2}, B_{0}, B_{1}, B_{2}\}$.

\begin{lemma}\label{lm:f3-otherSubfile}
    If $Z_0=\{A_0, A_1, B_2 \}$, then the permissible cache for $Z_1$
    must be one of the following. 
\begin{subequations}
\begin{eqnarray}
 Z_1=\left\{G_0, G_1 , A_2 \mid 
G_0, G_1\in \{B_0, B_1, B_2 \}\right\}\label{eq:f3-bothB}\\
\mbox{ or } Z_1=\left\{G_0, G_1 , A_2 \bigg|  \begin{array}{l}G_0\in \{A_0, A_1 \}\\   
G_1\in \{B_0, B_1, B_2 \} \end{array}\right\}.\label{eq:f3-oneAoneB}
\end{eqnarray}.
\end{subequations}
\end{lemma}
\begin{proof}
Consider the transmission
\begin{align}
    &X^{BA}=\left[\begin{array}{c} u\\v \end{array}\right]=\label{eq:f3-XBA}\\
    & \left[\begin{array}{c}\alpha_0 A_0 + \alpha_1 A_1+ \alpha_2 A_2+\beta_0 B_0+\beta_1 B_1+\beta_2 B_2 \\
    \gamma_0 A_0+\gamma_1 A_1+\gamma_2 A_2+\delta_0 B_0+\delta_1 B_1+\delta_2 B_2\end{array}\right]\notag
\end{align} 

User~0 can use its cache contents to eliminate three variables from the system of linear equations in Eq.~\eqref{eq:f3-XBA}. The reduced/equivalent equations for User~0 is 
\begin{align*}
    \left[\begin{array}{c} u'\\v' \end{array}\right]=& \left[\begin{array}{c}\alpha_2 A_2+\beta_0 B_0+\beta_1 B_1 \\
    \gamma_2 A_2+\delta_0 B_0+\delta_1 B_1\end{array}\right]
\end{align*}
Since User~0 does not have access to $A_2$, for recovering $B_0$ and $ B_1$ we need
\begin{align}
    \rk\left(\left[\begin{array}{cc}\beta_0 & \beta_1 \\
    \delta_0 & \delta_1 \end{array}\right]\right) =2 \mbox{ and }\left[\begin{array}{c}\alpha_2  \\
    \gamma_2\end{array}\right] = \left[\begin{array}{c}0  \\ 0 \end{array}\right]\label{eq:f3-xba-noA2-B01rank2} 
\end{align}
The transmission $X^{BA}$ cannot involve $A_2$. So $A_2$ must be in $Z_1$ for it to recover file $A$ from $X^{BA}$. 
\[G_2=A_2\]
All the subfiles in $Z_1$ cannot be that of file $A$. So, we have two cases for the possible values of $\{G_0, G_1\}$ based on the associated files. Either both of them are subfiles of $B$ as in Eq.~\eqref{eq:f3-bothB} or one of them is subfile of $A$ and the other is of $B$ as in Eq.\eqref{eq:f3-oneAoneB}.
\end{proof}

\subsection{Two subfiles of file $B$ in $Z_1$}
In this section, we will show that if the cache of User~1 is of the form given in Eq.~\eqref{eq:f3-bothB},
then the scheme is not private. 
\begin{lemma}\label{lm:f3-no2B}
    If $Z_0=\{A_0, A_1, B_2 \}$, and 
    $Z_1=\{G_0, G_1,  A_2 \}$ and $G_i\in \{B_0, B_1, B_2 \}$,
    then demand privacy is not satisfied. 
\end{lemma}
\begin{proof}
Let $G_0, G_1 \in \{B_0, B_1, B_2\}$.
The reduced equations corresponding to $X^{BA}$ for User~1 will be 
\begin{align}
    \left[\begin{array}{c} u''\\v'' \end{array}\right]= \left[\begin{array}{c}\alpha_0 A_0 + \alpha_1 A_1+ \beta_0 B_0+\beta_1 B_1+\beta_2 B_2 \\
    \gamma_0 A_0+\gamma_1 A_1+\delta_0 B_0+\delta_1 B_1+\delta_2 B_2\end{array}\right]
\end{align}

For User~1 being able to obtain subfiles $A_0$ and $A_1$ from the transmission, we need
\begin{subequations}
\begin{align}
    \rk\left(\left[\begin{array}{cc}\alpha_0 & \alpha_1 \\
    \gamma_0 & \gamma_1 \end{array}\right]\right) &=2, \mbox{ and }\\
    \left[\begin{array}{c}\beta_2  \\
    \delta_2\end{array}\right] &= \left[\begin{array}{c}0  \\ 
    0 \end{array}\right]
\end{align}
\end{subequations}
Due to Eq.~\eqref{eq:f3-xba-noA2-B01rank2}, $Z_2$ must contain the subfiles $B_0$ and $B_1$, for eliminating those variables from $X^{BA}$ and recover the subfiles of $A$.  
\begin{align}
    \{G_0, G_1\} &= \{B_{0}, B_{1}\}
\end{align}

Now consider the transmission for $D=(A,A)$. 
\begin{align}
    &X^{AA}=\left[\begin{array}{c} u\\v \end{array}\right]=\\& \left[\begin{array}{c}\alpha_0' A_0 + \alpha_1' A_1+ \alpha_2' A_2+\beta_0' B_0+\beta_1' B_1+\beta_2' B_2 \\
    \gamma_0' A_0+\gamma_1' A_1+\gamma_2' A_2+\delta_0' B_0+\delta_1' B_1+\delta_2' B_2\end{array}\right]\notag
\end{align}
For User~1, reduced equations are
\begin{align}
    \left[\begin{array}{c} u''\\v'' \end{array}\right]=\left[\begin{array}{c}\alpha_0' A_0 + \alpha_1' A_1+\beta_2' B_2 \\
    \gamma_0' A_0+\gamma_1' A_1+\delta_2' B_2\end{array}\right]
\end{align}
For User~1 recovering $A_0$ and $A_1$, it requires
\begin{subequations}
\begin{align}
    \rk\left(\left[\begin{array}{cc}\alpha_0' & \alpha_1' \\
    \gamma_0' & \gamma_1' \end{array}\right]\right) &=2\label{eq:f3-A01rank2} \mbox{ and }\\
    \left[\begin{array}{c}\beta_2'  \\
    \delta_2'\end{array}\right] &= \left[\begin{array}{c}0  \\ 0 \end{array}\right]
\end{align}
\end{subequations}
For demand privacy we require the existence of some cache $Z_1'$ which can recover file $B$ from $X^{AA}$. Since $X^{AA}$ doesn't involve $B_2$, it must be present in $Z'$.
\[Z_1'=\{H_0, B_2, A_2\}\]
For no value of $H_0 \in \{B_0, B_1, A_0, A_1\}$, it can recover both $B_0$ and $B_1$ (or file B completely) from $X^{AA}$ due to Eq.~\eqref{eq:f3-A01rank2}. Thus, if $Z_0$ has two subfiles of $A$, $Z_1$ cannot contain two subfiles of B as given in Eq.~\eqref{eq:f3-bothB}.
\end{proof}

\subsection{Two subfiles of file $A$ in $Z_1$}
If the cache of User~1 is of the form given in Eq.~\eqref{eq:f3-oneAoneB}, then we can restrict the 
cache even further as the following lemma shows. 
\begin{lemma}\label{lm:z1-no-b2}
     If $Z_0=\{A_0, A_1, B_2 \}$, then the permissible cache for $Z_1$
    must be of the form $Z_1=\{G_0, G_1, A_2 \}$, where $G_0, G_1$
    are distinct and $G_0 \in \{A_0, A_1\}$ and $G_1\in \{ B_0, B_1 \}$.
\end{lemma}
\begin{proof}
We need to show $G_1\neq B_2$.
Since $Z_1$ already contains $A_2$, it can have either $A_0$ or $A_1$, both of which are in $Z_0$. Without loss of generality, let $G_0 = A_1$. 
Assume $G_1 = B_2$. Then $Z_1=\{B_2, A_1, A_2\}$
Consider the transmission
\begin{align}
    &X^{AB}=\left[\begin{array}{c} u\\v \end{array}\right]=\\& \left[\begin{array}{c}\alpha_0 A_0 + \alpha_1 A_1+ \alpha_2 A_2+\beta_0 B_0+\beta_1 B_1+\beta_2 B_2 \\
    \gamma_0 A_0+\gamma_1 A_1+\gamma_2 A_2+\delta_0 B_0+\delta_1 B_1+\delta_2 B_2\end{array}\right]\notag
\end{align}
For User~0, these equations reduces to
\begin{align}
    \left[\begin{array}{c} u'\\v' \end{array}\right]=\left[\begin{array}{c}\alpha_2 A_2+\beta_0 B_0+\beta_1 B_1 \\
    \gamma_2 A_2+\delta_0 B_0+\delta_1 B_1 \end{array}\right]
\end{align}
Since User~0 have no access to $B_0$ and $B_1$, for recovering $A_2$, we need
\begin{subequations}
\begin{align}
    \rk\left(\left[\begin{array}{cc}\beta_0 & \beta_1 \\
    \delta_0 & \delta_1 \end{array}\right]\right) &\leq 1\label{eq:f3-XAB-U0-B01rank1} \mbox{ and }\\
    \left[\begin{array}{c}\alpha_2  \\
    \gamma_2\end{array}\right] &\neq \left[\begin{array}{c}0  \\ 0 \end{array}\right] 
\end{align}
\end{subequations}
For User~1 the equations from $X^{AB}$ reduces to
\begin{align}
    \left[\begin{array}{c} u''\\v'' \end{array}\right]= \left[\begin{array}{c}\alpha_0 A_0 + \beta_0 B_0+\beta_1 B_1 \\
    \gamma_0 A_0+ \delta_0 B_0+\delta_1 B_1\end{array}\right]
\end{align}
For User~1 recovering $B_0$ and $B_1$, it requires
\begin{subequations}
\begin{align}
    \rk\left(\left[\begin{array}{cc}\beta_0 & \beta_1 \\
    \delta_0 & \delta_1 \end{array}\right]\right) &=2\label{eq:f3-XAB-U1-B01rank2} \mbox{ and }\\
    \left[\begin{array}{c}\alpha_0  \\
    \gamma_0\end{array}\right] &= \left[\begin{array}{c}0  \\ 0 \end{array}\right] 
\end{align}
\end{subequations}
Equations~\eqref{eq:f3-XAB-U1-B01rank2} and \eqref{eq:f3-XAB-U0-B01rank1} are contradictory, Hence $G_1 \neq B_2$. So $G_1 \in \{B_0, B_1\}$. 
\end{proof}

By Lemma~\ref{lm:z1-no-b2}, there are four possible choices for $Z_1$ as given below.
\begin{subequations}
\begin{align}
Z_a=\{A_1, A_2, B_0\}\\
Z_b=\{A_1, A_2, B_1\}\\
Z_c=\{A_0, A_2, B_0\}\\
Z_d=\{A_0, A_2, B_1\}
\end{align}
\end{subequations}

\begin{lemma}\label{lm:no-dp-2}
If $Z_0=\{A_0, A_1, B_2 \}$ and $Z_1\in \{Z_a,Z_b,Z_c,Z_d\}$, then demand privacy is not possible. 
\end{lemma}
\begin{proof}
It suffices to show demand privacy is not possible for $Z_1=Z_a$, since we can arrive at the other cache combinations by relabeling.

Suppose the cache $Z_a=\{A_1, A_2, B_0\}$ is assigned to the User~1. Then, by arguments similar to Lemmas~\ref{lm:f3-otherSubfile},\ref{lm:f3-no2B} and \ref{lm:z1-no-b2}, the User~1 is aware that the cache of User~0 must have two subfiles of $A$, with one being $A_0$ and the subfile of $B$ in $Z_0$ is not $B_0$. The four possible caches for $Z_0$ consistent with $Z_1=Z_a$ are given below. 
\begin{subequations}
\begin{align}
Z_e=\{A_0, A_1, B_1\}\\
Z_f=\{A_0, A_1, B_2\}\\
Z_g=\{A_0, A_2, B_1\}\\
Z_h=\{A_0, A_2, B_2\}
\end{align}
\end{subequations}
Note that $Z_0=Z_f$. For demand privacy we need the caches consistent with $Z_0$ to be able to recover both files and vice versa.

Consider the transmission
\begin{align}
    &X^{AB}=\left[\begin{array}{c} u\\v \end{array}\right]=\\& \left[\begin{array}{c}\alpha_0 A_0 + \alpha_1 A_1+ \alpha_2 A_2+\beta_0 B_0+\beta_1 B_1+\beta_2 B_2 \\
    \gamma_0 A_0+\gamma_1 A_1+\gamma_2 A_2+\delta_0 B_0+\delta_1 B_1+\delta_2 B_2\end{array}\right]\notag
\end{align}

For the User~0 with cache $Z_0=\{A_0, A_1, B_2\}$ the transmission $X^{AB}$  reduces to the following set of equations after eliminating the subfiles which are already present in $Z_0$.
\begin{align}
    \left[\begin{array}{c} u_0\\v_0 \end{array}\right]=\left[\begin{array}{c}\alpha_2 A_2+\beta_0 B_0+\beta_1 B_1 \\
    \gamma_2 A_2+\delta_0 B_0+\delta_1 B_1 \end{array}\right]
\end{align}
For User~0 whose demand is $A$ already has $A_0$ and $A_1$. Only $A_2$ needs to be  recovered from 
$X^{AB}$.
This is possible only if the following conditions are satisfied. 
\begin{subequations}
\begin{align}
    \rk\left(\left[\begin{array}{cc}\beta_0 & \beta_1 \\
    \delta_0 & \delta_1 \end{array}\right]\right) &\leq 1 \label{eq:f3-xab-b01rank2} \mbox{ and }\\
    \left[\begin{array}{c}\alpha_2  \\
    \gamma_2\end{array}\right] &\neq \left[\begin{array}{c}0  \\ 0\label{eq:f3-xab-a2not0} \end{array}\right] 
\end{align}
\end{subequations}

Similarly, for the User~1, whose cache is $Z_1=Z_a$, the transmission $X^{AB}$ reduces to
\begin{align}
    \left[\begin{array}{c} u_1\\v_1 \end{array}\right]=\left[\begin{array}{c}\alpha_0 A_0+\beta_1 B_1+\beta_2 B_2 \\
    \gamma_0 A_0+\delta_1 B_1+\delta_2 B_2 \end{array}\right]
\end{align}
For User~1 to recover $B_0$ and $B_1$, the following conditions must be satisfied.
\begin{subequations}
\begin{align}
    \rk\left(\left[\begin{array}{cc}\beta_1 & \beta_2 \\
    \delta_1 & \delta_2 \end{array}\right]\right) &=2 \label{eq:f3-xab-b12rank2}\mbox{ and }\\
    \left[\begin{array}{c}\alpha_0  \\
    \gamma_0\end{array}\right] &= \left[\begin{array}{c}0 \label{eq:f3-xab-noA0} \\ 0 \end{array}\right] 
\end{align}
\end{subequations}

The reduced equations for $Z_b$ are
\begin{align}
    \left[\begin{array}{c} u_b\\v_b \end{array}\right]=\left[\begin{array}{c}\alpha_0 A_0+\beta_0 B_0+\beta_2 B_2 \\
    \gamma_0 A_0+\delta_0 B_0+\delta_2 B_2 \end{array}\right]
\end{align}

The reduced equations for $Z_c$ are
\begin{align}
    \left[\begin{array}{c} u_c\\v_c \end{array}\right]=\left[\begin{array}{c}\alpha_1 A_1+\beta_1 B_1+\beta_2 B_2 \\
    \gamma_1 A_1+\delta_1 B_1+\delta_2 B_2 \end{array}\right]
\end{align}

The reduced equations for $Z_d=Z_g$ are
\begin{align}
    \left[\begin{array}{c} u_d\\v_d \end{array}\right]=\left[\begin{array}{c}\alpha_1 A_1+\beta_0 B_0+\beta_2 B_2 \\
    \gamma_1 A_1+\delta_0 B_0+\delta_2 B_2 \end{array}\right]
\end{align}

The reduced equations for $Z_e$ are
\begin{align}
    \left[\begin{array}{c} u_e\\v_e \end{array}\right]=\left[\begin{array}{c}\alpha_2 A_2+\beta_0 B_0+\beta_2 B_2 \\
    \gamma_2 A_2+\delta_0 B_0+\delta_2 B_2 \end{array}\right]
\end{align}

The reduced equations for $Z_h$ are
\begin{align}
    \left[\begin{array}{c} u_h\\v_h \end{array}\right]=\left[\begin{array}{c}\alpha_1 A_1+\beta_0 B_0+\beta_1 B_1 \\
    \gamma_1 A_1+\delta_0 B_0+\delta_1 B_1 \end{array}\right]
\end{align}

From the above constraints and the reduced equations for all users we can infer the following.
\begin{enumerate}
    \item Due to  Eq.~\eqref{eq:f3-xab-noA0}, $Z_b$ cannot recover file $A$ since it has no access to $A_0$.
    \item Due to Eq.~\eqref{eq:f3-xab-b12rank2}, and since $Z_c$ has no access to $B_1$ and $B_2$, it cannot cannot eliminate them from the transmission to recover file $A$.
    \item Due to Eq.~\eqref{eq:f3-xab-a2not0}, $Z_e$ cannot recover file $B$.
    \item Due to Eq.~\eqref{eq:f3-xab-b01rank2}, $Z_h$ cannot recover file $B$.
\end{enumerate}

Hence, from the four possible caches for $Z_1$, the only cache that might be able to recover file $A$ and might achieve privacy for User~1 is $Z_d$. But since there are only five equations from the cache and transmissions, it is impossible for $Z_d$ to recover all the six subfiles $A_0, A_1, A_ 2, B_0, B_1, B_2$ and thus recover file $B$ also. That means, no possible cache for User~0 consistent with $Z_1$ is able to recover file $B$ from $X^{AB}$. This results in no demand privacy for User~0. 

On the other hand, if $X^{AB}$ is such that $Z_d$ can recover file $B$, then it results in no privacy for User~1. 
\end{proof}
Note that any consistent set of caches for User~0 and User~1 can be obtained by permuting subfile labels of file $A$ and file $B$ (permutation $\pi_A$ to relabel $A$ and $\pi_B$ to relabel $B$). Applying the same  relabeling, the above proof will hold true for them as well. This concludes the proof of Lemma~\ref{lem:nosub3noncoded}.
\end{document}